\documentclass[12pt]{article}
\usepackage{amsmath,amsthm,txfonts,graphicx,subfigure}
\usepackage[T1]{fontenc}
\usepackage{fullpage}

\newcommand{\cng}{\mathrm{cng}}
\newcommand{\stc}{\mathrm{stc}}

\newcommand{\Bndry}{\theta} 
\newcommand{\In}{\iota}     

\newcommand{\ceil}[1]{\left\lceil #1 \right\rceil}
\newcommand{\floor}[1]{\left\lfloor #1 \right\rfloor}

\newcommand{\theoremname}{Theorem}
\newtheorem{thm}{\theoremname}[section]
\newcommand{\lemname}{Lemma}
\newtheorem{lem}[thm]{\lemname}
\newcommand{\corname}{Corollary}

\newcommand{\propname}{Proposition}

\newcommand{\claimname}{Claim}

\newcommand{\ead}[1]{\texttt{#1}}

\title{On spanning tree congestion of Hamming graphs}
\author{
Kyohei Kozawa\thanks{
Electric Power Development Co., Ltd.,
6-15-1, Ginza, Chuo-ku, Tokyo, 104-8165 Japan.}
\and
Yota Otachi\thanks{
Graduate School of Information Sciences, Tohoku University,
Sendai 980-8579, Japan.
E-mail address: \ead{otachi@dais.is.tohoku.ac.jp}}
}

\begin{document}
\maketitle

 \begin{abstract}
  We present a tight lower bound for the spanning tree congestion of Hamming graphs.
 \end{abstract}

 \section{Preliminaries}
 \label{sec:intro}
 The spanning tree congestion of graphs has been studied intensively~\cite{BFGOL2011stc,BKMO2011,CO09,Hru08,KO2011rook,KOY09stc,Law09,LO2010stc,LO2011stc,LRR09,OOUU2011tamc,Ost04,Ost2010,Ost2011stc,Sim87}.
 In this note, we study the spanning tree congestion of Hamming graphs.
 We present a lower bound for the spanning tree congestion of Hamming graphs.
 That is, in our terminology, we show that $\stc(K_{n}^{d}) \ge \frac{1}{d}(n^{d}-1) \log_{n} d$.
 It is known that $\stc(K_{n}^{d}) = O\left(\frac{1}{d} n^{d} \log_{n} d\right)$~\cite{LO2011stc}.
 Thus our lower bound is asymptotically tight.

 For a graph $G$, we denote its vertex set and edge set by $V(G)$ and $E(G)$, respectively.
 For $S \subseteq V(G)$, let $G[S]$ denote the subgraph induced by $S$.
 For an edge $e \in E(G)$, we denote by $G - e$ the graph obtained from $G$ by deleting $e$.
 Let $N_{G}(v)$ denote the neighborhood of $v \in V(G)$ in $G$; that is, $N_{G}(v) = \{u \mid \{u,v\} \in E(G)\}$.
 We denote the degree of a vertex $v \in V(G)$ by $\deg_{G}(v)$,
 and the maximum degree of $G$ by $\Delta(G)$; that is,
 $\deg_{G}(v) = |N_{G}(v)|$ and $\Delta(G) = \max_{v \in V(G)} \deg_{G}(v)$.
 A graph $G$ is \emph{$r$-regular} if $\deg_{G}(v) = r$ for every $v \in V(G)$.

 For $S \subseteq V(G)$,
 we denote the edge set of $G[S]$ by $\In_{G}(S)$, and 
 the \emph{boundary edge set} by $\Bndry_{G}(S)$; that is,
 $\In_{G}(S) = \{\{u,v\} \in E(G) \mid u, v \in S\}$ and
 $\Bndry_{G}(S) = \{\{u,v\} \in E(G) \mid \textrm{exactly one of } u,v \textrm{ is in } S \}$.
 We define the function $\In$ and $\Bndry$ also for a positive integer $s \le |V(G)|$ as
 $\In_{G}(s) = \max_{S \subseteq V(G), \ |S| = s} |\In_{G}(S)|$ and
 $\Bndry_{G}(s) = \min_{S \subseteq V(G), \ |S| = s} |\Bndry_{G}(S)|$.
 Let $T$ be a spanning tree of a connected graph $G$.
 The \emph{congestion} of $e \in E(T)$ as $\cng_{G}(e) = |\Bndry_{G}(L_{e})|$,
 where $L_{e}$ is the vertex set of one of the two components of $T - e$.
 The \emph{congestion of $T$ in $G$}, denoted by $\cng_{G}(T)$, is the maximum congestion over all edges in $T$.
 We define the \emph{spanning tree congestion} of $G$, denoted by $\stc(G)$, as
 the minimum congestion over all spanning trees of $G$.

 The \emph{$d$-dimensional Hamming graph $K_{n}^{d}$} is the graph with vertex set $\{0,\dots,n-1\}^{d}$
 in which two vertices are adjacent if and only if their corresponding $d$-dimensional vectors
 differ in exactly one place.
 It is evident that $K_{n}^{d}$ is $d(n-1)$-regular.
 The exact value of $\stc(K_{n}^{2})$ is known~\cite{KO2011rook}.
 Also, $\stc(K_{2}^{d})$ is determined asymptotically~\cite{Law09}.

 \section{The lower bound}
 Here, we present the lower bound. We need the following three lemmas.
 \begin{lem}
  [\cite{Bez99}]
  \label{lem:regular}
  If $G$ is $r$-regular and $S \subseteq V(G)$,
  then $|\Bndry_{G}(S)| = r |S| - 2|\In_{G}(S)|$.
 \end{lem}
 \begin{lem}
  [\cite{STV01}]
  \label{lem:ham_num_edges}
  Let $G$ be a subgraph of $K_{n}^{d}$.
  If $G$ has $s$ vertices and $t$ edges, then $2t \le (n-1) s \log_{n} s$.
 \end{lem}
 \begin{lem}
  [\cite{CO09,KOY09stc}]
  \label{lem:lb_deg}
  For any connected graph $G$,
  $\stc(G) \ge \min_{s = \ceil{(|V(G)| - 1)/\Delta(G)}}^{\floor{|V(G)|/2}} \Bndry(s)$.
 \end{lem}

 \begin{thm}
  \label{thm:multi_lb}
  $\stc(K_{n}^{d}) \ge (n^{d}-1)\log_{n}d/d$ for $n,d \ge 3$.
 \end{thm}
 \begin{proof}
  Since $K_{n}^{d}$ is $d(n-1)$-regular,
  \lemname{s}~\ref{lem:regular} and \ref{lem:ham_num_edges} imply that
  $\Bndry_{K_{n}^{d}}(s) \ge (n-1) s (d - \log_{n} s)$.
  Let $f(s) = (n-1)s(d-\log_{n}s)$ and $f'(s)$ be the derived function of $f(s)$.
  Then $f'(s) = (n-1)(d - 1/\ln n - \log_{n} s)$, and thus,
  $f(s)$ is increasing for
  $(n^{d} -1)/(d(n-1)) \le s \le n^{d - 1/\ln n}$
  and decreasing for
  $n^{d - 1/\ln n} \le s \le n^{d}/2$.
  Therefore,
  \begin{align*}
   \min_{s = \ceil{(n^{d}-1)/(d(n-1))}}^{\floor{n^{d}/2}} f(s)
   &=
   \min\left\{{f\left(\ceil{\frac{n^d-1}{d(n-1)}}\right), f\left(\floor{\frac{n^{d}}{2}}\right)}\right\}
   \\&\ge
   \min\left\{f\left(\frac{n^d-1}{d(n-1)}\right), f\left(\frac{n^{d}}{2}\right)\right\}
   \\&=
   \min\left\{\frac{n^{d}-1}{d}\left(d-\log_{n}\frac{n^{d}-1}{d(n-1)}\right), \frac{(n-1)n^{d}}{2}\left(d-\log_{n}\frac{n^{d}}{2}\right) \right\}
   \\&\ge
   \min\left\{\frac{n^{d}-1}{d} \log_{n} d, \frac{(n-1)n^{d}}{2}\log_{n} 2 \right\}.
  \end{align*}
  Thus, by \lemname~\ref{lem:lb_deg}, it holds that
  \[
  \stc(K_{n}^{d})
  \ge
  \min\left\{\frac{n^{d}-1}{d} \log_{n} d, \frac{(n-1)n^{d}}{2}\log_{n} 2 \right\}.
  \]
  By a simple calculation, we can see that $\frac{n^{d}-1}{d} \log_{n} d \le \frac{(n-1)n^{d}}{2}\log_{n} 2$ for $d = 2, 3$.
  Since $n^{d}-1 \le (n-1) n^{d}$ and $(\log_{n} d)/d \le (\log_{n} 2)/2$ for $d \ge 4$,
  the theorem holds.
 \end{proof}


\end{document}